\documentclass[a4paper,UKenglish,cleveref, autoref, thm-restate]{lipics-v2021}

\pdfoutput=1 %
\hideLIPIcs  %

\title{An Optimal Randomized Algorithm for Finding the Saddlepoint}

\author{Justin Dallant}{Department of Computer Science, Université libre de Bruxelles, Belgium} {Justin.Dallant@ulb.be}{https://orcid.org/0000-0001-5539-9037}{Supported by the French Community of Belgium via the funding of a FRIA grant.}%

\author{Frederik Haagensen%
}{Department of Computer Science, IT University of Copenhagen, Denmark}{haag@itu.dk}{https://orcid.org/0000-0002-4161-4442}{Supported by Independent Research Fund Denmark, grant 0136-00144B, ``DISTRUST'' project.}

\author{Riko Jacob}{Department of Computer Science, IT University of Copenhagen, Denmark}{rikj@itu.dk}{https://orcid.org/0000-0001-9470-1809}{}

\author{László Kozma}{Institut für Informatik, Freie Universität Berlin, Germany}{laszlo.kozma@fu-berlin.de}{https://orcid.org/0000-0002-3253-2373}{Supported by DFG grant
KO 6140/1-2.}

\author{Sebastian Wild}{Department of Computer Science, University of Liverpool, UK}{wild@liverpool.ac.uk}{https://orcid.org/0000-0002-6061-9177}{Supported by EPSRC grant EP/X039447/1.}

\authorrunning{J. Dallant,  F. Haagensen, R. Jacob, L. Kozma and S. Wild} %

\Copyright{Justin Dallant, Frederik Haagensen, Riko Jacob, László Kozma, and Sebastian Wild} %



\category{} %

\relatedversion{} %

\nolinenumbers %

\EventEditors{}
\EventNoEds{2}
\EventLongTitle{}
\EventShortTitle{arXiv 2023}
\EventAcronym{arXiv}
\EventYear{2023}
\EventDate{}
\EventLocation{}
\EventLogo{}
\SeriesVolume{}
\ArticleNo{}

\usepackage{mathtools}
\usepackage{xfrac}
\newcommand{\DeclareAutoPairedDelimiter}[3]{%
  \expandafter\DeclarePairedDelimiter\csname Auto\string#1\endcsname{#2}{#3}%
  \begingroup\edef\x{\endgroup
    \noexpand\DeclareRobustCommand{\noexpand#1}{%
      \expandafter\noexpand\csname Auto\string#1\endcsname*}}%
  \x}

\DeclareAutoPairedDelimiter\ceil{\lceil}{\rceil}
\DeclareAutoPairedDelimiter\floor{\lfloor}{\rfloor}
\DeclarePairedDelimiter\pars{(}{)}
\newcommand{\OO}[1]{O\pars*{#1}}
\newcommand{\OOmega}[1]{\Omega\pars*{#1}}
\renewcommand{\epsilon}{\varepsilon}
\newcommand{\E}{\mathbb{E}}
\DeclareMathOperator{\polylog}{polylog}
\DeclareMathOperator{\width}{width}
\DeclareMathOperator{\height}{height}

\usepackage{wclrscode}

\usepackage{mdframed}  %

\newcommand{\lk}[1]{}

\usepackage{thmtools}
\usepackage{thm-restate}

\usepackage{tikz}
\usepackage{algpseudocode}
\usepackage{bbm}
\usepackage{microtype}

\begin{document}

\maketitle

\begin{abstract}
A \emph{saddlepoint} of an $n \times n$ matrix is an entry that is the maximum of its row and the minimum of its column. Saddlepoints give the \emph{value} of a two-player zero-sum game, corresponding to its pure-strategy Nash equilibria; efficiently finding a saddlepoint is thus a natural and fundamental algorithmic task. 

For finding a \emph{strict saddlepoint} (an entry that is the strict maximum of its row and the strict minimum of its column) we recently gave an $\OO{n\log^*{n}}$-time algorithm, improving the $\OO{n\log{n}}$ bounds from 1991 of Bienstock, Chung, Fredman, Schäffer, Shor, Suri and of Byrne and Vaserstein. 

In this paper we present an optimal %
$\OO{n}$-time algorithm for finding a strict saddlepoint based on random sampling. Our algorithm, like earlier approaches, accesses matrix entries only via unit-cost binary comparisons. For finding a (non-strict) saddlepoint, we extend an existing lower bound to randomized algorithms, showing that the trivial $O(n^2)$ runtime cannot be improved even with the use of randomness. 

\end{abstract}

\section{Introduction}\label{sec1}

Given a matrix $A$ with entries from a set of comparable elements (e.g., $\mathbb{R}$ or $\mathbb{N}$), a \emph{saddlepoint} of $A$ is an entry that is the maximum of its row and the minimum of its column. A \emph{strict saddlepoint} is an entry that is the strict maximum of its row and the strict minimum of its column. It is easy to see that a strict saddlepoint, if it exists, must be unique. 

If $A$ is the payoff matrix of a two-player zero-sum game and the payoffs are pairwise distinct, then a saddlepoint of $A$ (if it exists) is necessarily strict and corresponds to the pure-strategy Nash equilibrium, giving the \emph{value} of the game, e.g., see~\cite[\S\,4]{gt}. Finding a strict saddlepoint efficiently is thus a fundamental algorithmic task. Saddlepoint computation also arises in continuous optimization (e.g., for image processing or machine learning), where a (global or local, exact or approximate) saddlepoint of a function $f(x,y)$ is sought, typically under additional structural assumptions on $f$, e.g., see~\cite{chambolle2011first, chambolle2016ergodic, lin2020near, nedic2009subgradient, razaviyayn2020nonconvex}. By contrast, our problem setting is discrete, and we make no assumptions on the input matrix $A$; the iterative methods developed in the above settings are thus, not applicable.  

Finding a saddlepoint (strict or not) of an $n$-by-$n$ matrix $A$ can easily be done in $O(n^2)$ time (Knuth~\cite[\S\,1.3.2]{Knuth1}), and a simple adversary argument (Llewellyn, Tovey, and Trick~\cite{Llewellyn1988}) shows that in the presence of duplicates, every deterministic algorithm that finds a saddlepoint must make $\Omega(n^2)$ comparisons in the worst case. 

Strict saddlepoints turn out to be algorithmically more interesting, and perhaps surprisingly, we can find a strict saddlepoint (or report non-existence) examining only a vanishingly small part of $A$.   
The first subquadratic algorithm for finding a strict saddlepoint was obtained in 1988 by Llewellyn, Tovey, and Trick~\cite{Llewellyn1988} with a runtime of $O(n^{\log_2{3}}) \subset O(n^{1.59})$.
In 1991, Bienstock, Chung, Fredman, Schäffer, Shor, Suri~\cite{Bienstock1991}, and independently, Byrne and Vaserstein~\cite{Byrne1991} found algorithms with runtime $O(n\log{n})$. Due to the implicit sorting step of these algorithms and the lack of improvements in three decades, it was natural to expect this bound to be best possible. 

Very recently, however, an algorithm with runtime $O(n\log^*{n})$ was obtained by the authors of this paper~\cite{dallant2024finding}, where $\log^*(\cdot)$ is the slowly growing \emph{iterated logarithm} function. The algorithm of~\cite{dallant2024finding} as well as all earlier algorithms are deterministic.  Bypassing the sorting barrier has raised the possibility of a linear-time algorithm that would match the natural lower bound: $2n-2$ comparisons are required to verify that a given entry is indeed a saddlepoint.  In this paper we give a randomized algorithm attaining this bound (up to constant factors).

\begin{theorem}\label{thm1}
Given an $n \times n$ matrix $A$, we can identify the strict saddlepoint of $A$, or report its non-existence, in $O(n)$ time with high probability. 
\end{theorem}

Our algorithm is Las Vegas, i.e., it always gives the correct answer, with the runtime guarantee holding with high probability. The existence of a deterministic $O(n)$-time algorithm remains open. %
In \S\S\,\ref{sec2}, \ref{sec:finding_pivot}, \ref{sec4} we describe the algorithm and its analysis, proving Theorem~\ref{thm1} and further extensions. 

In \S\,\ref{sec5} we prove a lower bound on the efficiency of \emph{randomized} algorithms for the general saddlepoint problem, showing that the trivial quadratic runtime cannot be improved even with randomization. 

\begin{theorem}
Every randomized comparison-based algorithm that finds a (non-strict) saddlepoint with probability at least $\sfrac{5}{6}$ must take $\Omega(n^2)$ expected time on some $n \times n$ matrix.
\end{theorem}

\section{The overall algorithm}\label{sec2}

Let $A$ be an $n \times n$ input matrix with pairwise distinct, comparable entries, where $A_{ij}$ is the entry in the $i$-th row and $j$-th column. Note that the assumption of distinctness is only for convenience of presentation, we comment later on how to remove this assumption and also extend our results to non-square matrices. 

Our approach is based on the following reduction step: if every row of $A$ contains an entry at least as large as $A_{ij}$ (and thus, $A_{ij}$ is a lower bound on the value of the saddle point), then we can delete each column $j'$ of $A$ with an entry $A_{ij'} < A_{ij}$ (because such a column could not contain the saddlepoint), and the strict saddlepoint of the matrix, if it exists, is preserved.
Indeed, if the deleted column $j'$ were to contain a strict saddlepoint $A_{kj'}$, then $A_{ij} > A_{ij'} \geq A_{kj'} \geq A_{kx} \geq A_{ij}$ would yield a contradiction; here $A_{kx}$ is the entry in row $k$ that is at least as large as $A_{ij}$, and the second and third inequalities hold due to $A_{kj'}$ being a saddlepoint. We call such an entry $A_{ij}$ a \emph{horizontal pivot}, if at least a quarter of the entries in row $i$ are smaller than $A_{ij}$ (allowing to remove the corresponding columns), see Figure~\ref{fig:pivot}. 

\begin{figure}[tbh]
\centering
\includegraphics[scale=0.37]{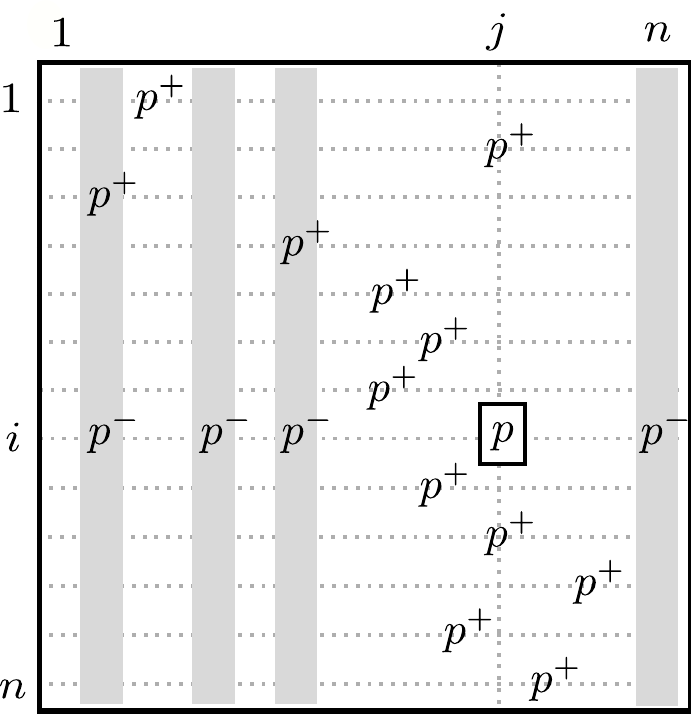}
\caption{\label{fig:pivot}Horizontal pivot $p = A_{ij}$ (framed). Entries denoted $p^{-}$ are smaller than $p$, entries denoted $p^{+}$ are larger than $p$. Columns marked gray cannot contain a strict saddlepoint.}
\end{figure}

Similarly, we call $A_{ij}$ a \emph{vertical pivot}, if every column has at least one entry at most as large as $A_{ij}$, and at least a quarter of the entries in column $j$ are larger than $A_{ij}$. By a symmetric argument, we can safely delete the rows of these entries. Note that these deletions may create a new, spurious strict saddlepoint; we can however, easily detect this later. 

Alternating in finding a horizontal and a vertical pivot and deleting a quarter of the columns and a quarter of the rows forms the core of our algorithm (Figure~\ref{fig1}). 
Notice that such pivots always exist, e.g., the minimum of all row-maxima is a horizontal pivot (it would allow the removal of all but one column), and the maximum of all column-minima is a vertical pivot; efficiently finding such entries, however, is far from obvious.  

\begin{figure}[h]
	\begin{codebox}
		\Procname{\proc{ReduceMatrix}$(A,s)$:}
		\li	\While $\height(A)>s$
		\Do
			\li \kw{try}
			\Do
				\li	$(i,j) \gets \proc{FindHorizontalPivot}(A)$
				\li Delete $\floor{\width(A)/4}$ columns $j'$ of $A$ with $A_{ij'} < A_{ij}$
				\li $(i,j) \gets \textsc{FindVerticalPivot}(A)$
				\li Delete $\floor{\height(A)/4}$ rows $i'$ of $A$ with $A_{i'j} > A_{ij}$
			\End
			\li \kw{catch} \textsc{Failed} pivot selection 
				\Then \li \Return \textsc{Failed}
			\End
	\End
		\li \Return $A$
	\end{codebox}
 \caption{Reducing the input matrix to size $s \times s$.\label{fig1}}
\end{figure}

Note that we could delete \emph{all} columns $j'$ with $A_{ij'} < A_{ij}$ (and all rows $i'$ with $A_{i'j} > A_{ij}$), but we restrict ourselves to deleting exactly a quarter of rows/columns to simplify the analysis. 

Assume for now that $\textsc{FindHorizontalPivot}$ (resp.\ $\textsc{FindVerticalPivot}$), when called on an $m\times m$ (or $m \times \lceil 3m/4 \rceil$) matrix, runs in $O(m)$ time, and returns a horizontal (resp.\ vertical) pivot with probability at least $1-f(m)$, for some decreasing function $f:\mathbb{N} \rightarrow [0,1]$ which will be made explicit later. 
Failure to find a suitable pivot will be reported as \textsc{Failed} pivot selection. 

\begin{theorem}\label{thm3}
    Let $A$ be an $n\times n$ matrix. %
    Then $\textsc{ReduceMatrix}(A,s)$ runs in $O(n)$ time and with probability at least $1-O(f(s)\log{\frac{n}{s}})$ returns an $s'\times s'$ submatrix $A'$ of $A$, with $s' \leq s$. If $A$ has a strict saddlepoint, then $A'$ has the same strict saddlepoint. %
\end{theorem}

\begin{proof}
The fact that the reduction steps preserve a strict saddlepoint is immediate from our preceding discussion. 
For an $m \times m$ matrix with $s \leq m \leq n$, the two pivot-finding calls take $O(m)$ time. Similarly, finding the columns and rows to be removed takes $O(m)$ time; this requires inspecting the row, resp.\ column of the pivots. 

The step succeeds with probability at least $1-2f(m) \geq 1-2f(s)$ (by the union bound), reducing the matrix to size $\lceil 3m/4 \rceil \times \lceil 3m/4 \rceil$. To delete rows and columns efficiently, only simple bookkeeping is needed: we maintain an array of the remaining row- and column-indices of the original matrix, compacting the array after each iteration, at amortized constant time per deletion.

Thus, starting with an $n \times n$ matrix, we obtain a matrix of size at most $s \times s$ in $O(\log{\frac{n}{s}})$ iterations, with failure probability (again by the union bound) of at most $O(f(s)\log{\frac{n}{s}})$. The total running time is $O(n +  n(3/4) + n(3/4)^2 + \cdots) = O(n)$. 
\end{proof}

In \S\,\ref{sec:finding_pivot} we show that the pivot-finding can be achieved with failure-probability $f(m) = e^{-\Omega(m^{1/20})}$. This yields the following.
\begin{theorem}
    Let $A$ be an $n\times n$ matrix with distinct values. Then we can find the strict saddlepoint of $A$ (or report non-existence) by a Las Vegas randomized algorithm that terminates after $O(n)$ time with probability at least $1-e^{-\Omega(n^{1/21})}$.
\end{theorem}
\begin{proof}
We run the reduction process of Theorem~\ref{thm3}, setting $s=n/\log_2 n$. The overall probability of success is at least $1 - O(f(n/\log_2{n}) \log\log{n}) \geq 1 -e^{-\Omega(n^{1/21})}$. To obtain a Las Vegas algorithm, we repeat the procedure until it succeeds. With the given probability, no repetition is necessary and the running time is $O(n)$.

Then, we run the deterministic $O(N\log{N})$-time strict saddlepoint algorithm of Bienstock et al.~\cite{Bienstock1991} on the resulting $N \times N$ matrix with $N \leq {n/\log_2{n}}$, in $O(n)$ total time. (Alternatively, the $O(N\log^*{N})$-time algorithm of~\cite{dallant2024finding} can also be used.) Recall that a strict saddlepoint of $A$ is preserved by the reduction. Thus, if the algorithm reports none, then $A$ has none. If the algorithm finds a strict saddlepoint of the reduced matrix, then it is either a strict saddlepoint of $A$ or a spurious one created by the reduction. This can be verified in $O(n)$ time, examining the row and column of the reported entry in $A$. 
\end{proof}

In \S\,\ref{sec4} we show that we can have $f(m) = \OO{n^{-1}}$, using only a total of $\OO{\log n}$ random bits. Again, setting $s=n/\log_2 n$, a similar argument yields the following.
\begin{theorem}\label{thm5}
    Let $A$ be an $n\times n$ matrix with distinct values. Then we can find the strict saddlepoint of $A$ (or report non-existence) by a Las Vegas randomized algorithm that uses only $O(\log n)$ random bits and terminates after $O(n)$ time with probability at least $1-\OO{\log{n}\log\log{n}/n}$.
\end{theorem}

Note that an algorithm which runs in $O(n)$ time with probability at least $1-g(n)$ can be turned into one that runs in $O(n)$ time with probability at least $1-g(n)^c$ for any integer~$c$ (by restarting at most $c$ times if the algorithm does not terminate within a given time budget). %

We also remark that our algorithm is easily parallelizable, and can be %
adapted in, say, the CREW PRAM model, to run with high probability in $O(\polylog{n})$ time and $O(n)$ total work; %
we give more details in \S\,\ref{sec4}. 
By contrast, the earlier deterministic $O(n \log{n})$-time 
algorithms~\cite{Bienstock1991, Byrne1991, dallant2024finding} are inherently sequential: they rely on \emph{adaptively} querying $n$ entries of the matrix, where the choice of each query depends on comparisons involving previously queried items. %

\section{Finding a pivot}\label{sec:finding_pivot}

In this section we describe and analyse the procedure for finding a pivot. We discuss only horizontal pivots as the case of vertical pivots is entirely symmetric. See Figure~\ref{fig2} for the description of the procedure. 
We rely on linear-time selection: \textsc{Select}$(X,i)$ returns the $i$-th smallest entry in $X$ in time $O(|X|)$, and on sampling \emph{with replacement}: 
each call to \textsc{Rand}$(k)$ returns an element drawn independently, uniformly at random from $\{1,\ldots,k\}$ in $O(1)$ time. In \S\,\ref{sec4} we clarify this assumption: all the necessary samples can be generated upfront in time $O(n)$ with a success probability of at least $1-e^{-\Omega(n)}$.  

\begin{figure}[!h]
	\begin{codebox}
	\Procname{\proc{FindHorizontalPivot}$(A)$:}
		\li	Let $R$ be the set of rows of $A${, each of length $k$}
		\li $m \gets |R|$
		\\[-.5\baselineskip]
		
		\zi \Comment{\textsc{(Phase 1)}}
		\li $t \gets \infty$
		\li \While $|R|>\floor{m^{19/20}}$
		\Do
            \li	\For {$i$ in $1,\ldots, |R|$}
    		\Do
                \li {$q_i \gets$ the \textsc{Rand}$(k)$-th element in the $i$-th row of $R$}
                
    		\End
			\li	$\mathcal{R} \gets \left\{q_i \mid 1\leq i \leq |R|\right\}$
			\li	$q \gets \textsc{Select}\bigl(\mathcal{R},\ceil{\frac34|\mathcal{R}|} \bigr)$
			\li	$t \gets \min\{t,q\}$
			\li	Delete from $R$ all rows $i$ where $q_i > t$
		\End
		\\[-.5\baselineskip]
		
		\zi \Comment{\textsc{(Phase 2)}}
		\li	\Foreach remaining row $r$ in $R$
		\Do
              \li	\For {$i$ in $1,\ldots, \floor{m^{1/20}}$}
    		\Do
                \li	{$x_i :=$ the \textsc{Rand}$(k)$-th element of row $r$}

    		\End

            \li {$\mathcal{R}_r \gets \left\{x_i \mid 1 \leq i \leq \floor{m^{1/20}} \right\}$}
			\li	$q'_r \gets \textsc{Select}\bigl(\mathcal{R}_r, \floor{\frac{2}{5}|\mathcal{R}_r|}\bigr)$
		\End
		\\[-.5\baselineskip]
		
		\li	$p \gets \min_r\{q'_r\}$
		\li	\If $p>t$ or $p$ is not larger than $\lfloor k/4 \rfloor$ entries in its row in $A$
		\Then
			\li \Return \textsc{Failed}
		\End
		\li \Else\Do
			\li	\Return $p$
		\End
	\end{codebox}
\caption{Finding a horizontal pivot of the input matrix.\label{fig2}}
\end{figure}

The intuition of our procedure is as follows.  
To find a likely candidate for a horizontal pivot $p$, we want to find a value $q_r$ in a row $r$ where (at least) a quarter of the elements in $r$ are smaller than $q_r$. 
By choosing $p$ as the minimal $q_r$ across all rows, we guarantee the second requirement (that every row contains an element larger than $p$).
For a single row, we can obtain a likely value $q$ from a random sample in sublinear time, but we cannot afford to repeat this for all rows.
Therefore, we first reduce the number of rows by guessing an upper bound $t$ for $q$ and removing all rows that contain some element larger than $t$; if the ultimate candidate for the pivot $p$ is indeed less than $t$, those discarded rows already satisfy the requirement to contain an entry larger than $p$ and were (with hindsight) justifiably removed. 

Before turning to correctness, let us argue that $\textsc{FindHorizontalPivot}(A)$ runs in $O(m+k)$ time, where $m$ and $k$ are the number of rows and columns of $A$ respectively. %
Each round of the while loop in Phase $1$ runs in $O(|R|)$ time, where $R$ is the current set of rows, and $|R|$ decreases by at least a constant fraction each time, leading to a geometric series bounded by $O(m)$ overall. In Phase $2$, we spend $O(m^{1/20})$ time per row, and there are $O(m^{19/20})$ remaining rows, requiring $O(m)$ time overall. Finally, checking that $p$ is indeed a valid pivot requires looking at its row in $A$, in time $O(k)$.

\subsection{Correctness}

We will make use of the following tail bounds multiple times.
\begin{lemma}[Multiplicative Chernoff~{\cite[Thm.\,1.1]{dubhashi2009concentration}}]\label{lemma:chernoff}
    Let $X_1, \ldots, X_m$ be a sequence of independent Bernoulli random variables (with possibly distinct success probabilities). Let $X = \sum_{i=1}^m X_i$ and $\mu = \E(X)$. Then, for any constant $\epsilon > 0$, as $m\to\infty$,
    \[\Pr[X \geq (1+\epsilon)\mu] \;\leq\; e^{-\Omega(\mu)},\]
    and 
    \[\Pr[X \leq (1-\epsilon)\mu] \;\leq\; e^{-\Omega(\mu)}.\]
\end{lemma}

Let $R$ be a set of rows. %
For each row $r$ of $R$, let $X_r$ be a random variable distributed uniformly at random over the entries of $r$. Let $t>0$ be some threshold value and let $Q = X_{(\lceil3|R|/4\rceil)}$ be the sample third quartile of the set $\{X_r\}_{r\in R}$.

Our goal is to show that in each iteration of Phase 1 in \textsc{FindHorizontalPivot}, the set of rows we keep based on a single random sample from the row ``sufficiently resembles'' the set of rows that would be kept by a deterministic procedure based on the median. We proceed with two lemmas that correspond to the two cases of Line 9, Figure~\ref{fig2}.

\begin{lemma}\label{lemma:phase_1a}
    Let $m$ be an integer, and let $S \subseteq R$ be the set of rows in $R$ whose median is at most $t$. Suppose $|S| \geq \OOmega{m^{2/3}}$. Let $S' = \{s\in S \mid X_s \leq t\}$. Then with probability at least $1-e^{-\OOmega{m^{2/3}}}$, we have $|S'|\geq |S|/5$.
\end{lemma}
\begin{proof}
    For $s\in S$, let %
    $Y_s = 1$ if $X_s \leq t$ and $Y_s=0$ otherwise. Let $Y = \sum_{s\in S}Y_s$. We have $|S'| = Y$, and $\E(Y) = \sum_{s\in S}\E(Y_s) \geq |S|/2$.

    The variable $Y$ is a sum of independent Bernoulli variables. We can thus apply the Chernoff bound of Lemma \ref{lemma:chernoff}:
    \begin{align*}
        \Pr[Y\leq |S|/5] &\leq \Pr[Y\leq (1-3/5)\E(Y)]\\
        &\leq e^{-\OOmega{\E(Y)}}\\
        &\leq e^{-\OOmega{m^{2/3}}}.
    \end{align*}
The claim thus holds.
\end{proof}

\begin{lemma}\label{lemma:phase_1b}
    Suppose $|R| \geq \OOmega{m^{2/3}}$. Let $S \subseteq R$ be the set of rows in $R$ whose median is at most~$Q$. Let $S' = \{s\in S \mid X_s \leq Q\}$. Then with probability at least $1-e^{-\OOmega{m^{2/3}}}$, $|S'|\geq |S|/5$.
\end{lemma}
\begin{proof}
    For $r\in R$, let $Y_r = 1$ if $X_r \geq m_r$ and $Y_r=0$ otherwise, where $m_r$ denotes the median value of $r$. Let $Y = \sum_{r\in R}Y_r$. We have $\E(Y) = \sum_{r\in R}\E(Y_r) \geq |R|/2$.

    The variable $Y$ is a sum of independent Bernoulli variables. We can thus again apply the Chernoff bound of Lemma \ref{lemma:chernoff}:
    \begin{align*}
        \Pr\left[Y\leq \frac{9}{20}|R|\right] &\leq \Pr\left[Y\leq \left(1-\frac{1}{10}\right)\E(Y)\right]\\
        &\leq e^{-\OOmega{\E(Y)}}\\
        &\leq e^{-\OOmega{m^{2/3}}}.
    \end{align*}
Thus, with probability at least $1-e^{-\OOmega{m^{2/3}}}$ there are at least $\frac{9}{20}|R|$ rows in $R$ whose median is below the sampled value. Among those, at most $\frac{1}{4}|R|$ have a sampled value greater than $Q$ (by definition of $Q$). We conclude that at least $\frac{9}{20}|R|-\frac{1}{4}|R| = \frac{1}{5}|R|$ rows have a sampled value above their median and below $Q$. These rows are all in $S'$, therefore $|S'| \geq \frac{1}{5}|R| \geq \frac{1}{5}|S|$.
\end{proof}

\begin{lemma}
    Let $S$ be the set of rows in $R$ whose median is at most $\min\{t,Q\}$. Suppose $|S| \geq \OOmega{m^{2/3}}$. Let $S' = \{s\in S \mid X_s \leq \min\{t,Q\}\}$. Then, with probability at least $1-e^{-\OOmega{m^{2/3}}}$, $|S'|\geq |S|/5$.
\end{lemma}
\begin{proof}
This is immediate from the previous two lemmas.
\end{proof}

Because in each round of the while loop in Phase $1$, a quarter of all rows are deleted, the loop runs at most $N=\frac{1}{20}\log_{4/3}(m) + O(1)$ times. By the previous lemma (together with a union bound), we get the following.

\begin{proposition}\label{prop:phase1}
    At the end of the $i$-th iteration of the while loop in Phase $1$, with probability at least $1-i\cdot e^{-\OOmega{m^{2/3}}}$, at least $m\cdot(\frac{1}{5})^i$ of the rows in $R$ have median at most $t$. 
    
    In particular, after the last iteration of the loop, with probability at least $1-N\cdot e^{-\OOmega{m^{2/3}}} = 1- e^{-\OOmega{m^{2/3}}}$, at least $m\cdot(\frac{1}{5})^N \geq \OOmega{m\cdot5^{-\frac{1}{20}\log_{4/3}(m)}} \geq \OOmega{m^{2/3}}$ of the rows in $R$ have median at most $t$.
\end{proposition}

In Phase $2$ we aim to pick an element from each remaining row that is simultaneously below the median of the row (and thus, with some probability, below the threshold $t$), and above a quarter of the elements of the row (and thus, a good candidate for being a horizontal pivot). The following lemma ensures this. 
\begin{lemma}
    Let $k$ be an integer, and let $r$ be a row of $k$ distinct values. Sample $c=\floor{m^{1/20}}$ entries of row $r$  uniformly at random, with replacement: $Y_1,\ldots, Y_c$. Let $Y$ be the $\floor{\frac{2}{5}m^{1/20}}$-th order statistic of $\{Y_1,\ldots, Y_c\}$. With probability at least $1-e^{-\OOmega{m^{1/20}}}$, $Y$ is between the $\floor{k/4}$-th smallest element and the median of row $r$.
\end{lemma}
\begin{proof}
    Let $\ell_r$ be the $\floor{k/4}$-th smallest element of row $r$ and let $m_r$ be the median of row $r$.
    Let $Z_\ell$ be the number of variables among $Y_1,\ldots, Y_c$ which are smaller than $\ell_r$, and let $Z_m$ be the number of variables among $Y_1,\ldots, Y_c$ which are larger than $m_r$.
    Both $Z_\ell$ and $Z_m$ can be represented as sums of independent Bernoulli variables and have respective expectation $\E(Z_\ell) \leq \frac{1}{4}m^{1/20}$ and $\E(Z_m) \leq \frac{1}{2}m^{1/20}$.
    By the Chernoff bound of Lemma \ref{lemma:chernoff}:
    \begin{align*}
        \Pr\left[Z_\ell \geq \frac{2}{5}m^{1/20}\right] &\leq \Pr\left[Z_\ell \geq \left(1+\frac{3}{5}\right)\E(Z_\ell)\right]\\
        &\leq e^{-\OOmega{\E(Z_\ell)}}\\
        &\leq e^{-\OOmega{m^{1/20}}}.
    \end{align*}

    Similarly, we have:
    \begin{align*}
        \Pr\left[Z_{m} \geq \frac{3}{5}m^{1/20}\right] & \leq \Pr\left[Z_{m} \geq \left(1+\frac{1}{5}\right)\E(Z_{m})\right]\\
        &\leq e^{-\OOmega{m^{1/20}}}.
    \end{align*}

    By the union bound, the probability that $Z_\ell \geq \frac{2}{5}m^{1/20}$ or $Z_{m} \geq \frac{3}{5}m^{1/20}$ is at most $e^{-\OOmega{m^{1/20}}}$. This implies that with probability at least $1-e^{-\OOmega{m^{1/20}}}$, $Y$ is between $\ell_r$ and $m_r$, which proves the claim.
\end{proof}

Using the previous lemma together with a union bound across the ${{m^{19/20}}}$ remaining rows, we get the following.
\begin{proposition}\label{prop:phase2}
    In Phase 2, with probability at least $1-e^{-\OOmega{m^{1/20}}}$, the selected value for every row is between its $\floor{k/4}$-th smallest value and its median, where $k$ is the width of $A$.
\end{proposition}

The correctness of the algorithm (with high probability) is now easy to establish.
\begin{theorem}
    Let $A$ be a matrix with $m$ rows. Then $\textsc{FindHorizontalPivot}(A)$ finds a horizontal pivot of $A$ with probability at least $1-e^{-\OOmega{m^{1/20}}}$. 
\end{theorem}
\begin{proof}
    Both phases of the algorithm succeed (i.e., the events described in Propositions \ref{prop:phase1} and \ref{prop:phase2} happen) with probability at least $1-e^{-\OOmega{m^{1/20}}}$, by union bound. Assuming they indeed have, let $p$ be the value of the entry returned by $\textsc{FindHorizontalPivot}(A)$. Because by Proposition \ref{prop:phase1} there is a row with median smaller than $t$, we know by Proposition \ref{prop:phase2} that $p < t$. By our choice of $t$, every row deleted in Phase 1 has a value larger than $t$. Moreover, by the definition of $p$, every row remaining in $R$ by the end of the algorithm has a value larger than $p$. Thus, every row of $A$ has at least one value larger than $p$, and at least a quarter of the values in the row of $p$ are smaller than $p$ (again, by Proposition \ref{prop:phase2}). In short, $p$ is a horizontal pivot of $A$.
\end{proof}

\section{Sampling, derandomization and further remarks}\label{sec4}

{As mentioned before, we assume $\textsc{Rand}(k)$ to return, on each call, a uniform random integer from $\{1,\dots,k\}$ in $O(1)$ amortized time. This can be justified as follows. For each call of \textsc{FindHorizontalPivot} with $m$ rows, \textsc{Rand} is called in Phase 1 at most $m + m(3/4) + m(3/4)^2 + \cdots \leq 4m$ times, and in Phase 2 at most $m^{1/20}$ times for each of the $m^{19/20}$ rows, for a total of $5m$ times. A similar analysis applies to \textsc{FindVerticalPivot}. 
As the number of rows and columns is initially $n$ and decreases geometrically for subsequent $\textsc{FindPivot}$ calls, the overall number of calls to \textsc{Rand} is $O(n)$, with its input $k$ always at most $n$. We assume therefore that a sequence $S$ of $O(n)$ uniform random integers on $\ceil{\log_2{n}}$ bits are available upfront. We implement \textsc{Rand}$(k)$ by standard rejection sampling, e.g., see~\cite[\S\,1]{motwani1995}: take the $\ceil{\log_2{k}}$ least significant bits of the next integer from $S$, rejecting the value if it is out of bounds (this happens with probability at most $1/2$). To ensure that we do not run out of random numbers with probability at least $1-e^{-\Omega(n)}$ (which leaves the overall analysis of our algorithm unaffected), we just need, by Lemma \ref{lemma:chernoff}, a small constant times as many samples in $S$ than foreseen.}

\subparagraph*{Derandomization.}

To reduce the amount of randomness used, we replace the $O(n)$ uniform random {integers} in $\{1,\ldots, 2^{\ceil{\log_2{n}}}\}$ needed for sampling by a set of $O(n)$ random {integers} that are $d$-wise independent for some sufficiently large $d \in O(1)$. These can be generated from $O(\log{n})$ uniform random bits in $O(n)$ time (assuming the word RAM model of computation), {using known techniques based on polynomials with random coefficients~\cite[\S\,3]{vadhan2012pseudorandomness}.} 

Note however that if the procedure \textsc{ReduceMatrix} fails and needs to be repeated, then we do need fresh randomness for our analysis to go through.

In the analysis of the algorithm with $d$-wise independent random {integers}, we make use of the following lemma, replacing the Chernoff bounds of the previous section {and of the rejection sampling}. 

\begin{lemma}[\cite{bellare1994randomness}]\label{lemma:tail_bound}
    Let $d>0$ be an even constant and let $\{X_1, \ldots, X_m\}$ be a set of $d$-wise independent Bernoulli random variables (with possibly distinct success probabilities). Let $X = \sum_{i=1}^m X_i$ and $\mu = \E(X)$. Then, for any constant $\epsilon > 0$, as $m\to\infty$,
    \[Pr[X \geq (1+\epsilon)\mu] \;\leq\; \OO{\left(\frac{md}{\mu^2}\right)^{d/2}},\]
    and 
    \[Pr[X \leq (1-\epsilon)\mu] \;\leq\; \OO{\left(\frac{md}{\mu^2}\right)^{d/2}}.\]
\end{lemma}

Following a similar reasoning as in \S\,\ref{sec4}, we obtain Theorem~\ref{thm5}, i.e., the running time remains unchanged, albeit with a decreased probability of success that converges to $1$ polynomially rather than exponentially. In this extended abstract we omit the detailed calculations. 

\subparagraph*{Avoiding the use of constant-time multiplication.} The standard technique mentioned above to generate $d$-wise independent random integers (for some constant $d$) relies on evaluating a polynomial of degree $d$ at inputs $x=0,1,\ldots, N$, for some $N \in O(n)$. Doing this na\"{i}vely makes use of $O(n)$ multiplications and additions. If we assume, as is often done, that the cost of multiplying word-sized integers is constant, then this is within the stated time bounds. We note however that it is possible to get around this assumption, using a technique by Knuth~\cite[\S\,4.6.4]{Knuth2} for evaluating a polynomial at the points along an arithmetic progression. By doing so, we trade the $O(n)$ multiplications and additions for $O(1)$ multiplications and $O(n)$ additions. Thus, even if multiplication is not a constant-time primitive and is instead implemented as, say, binary long multiplication, the stated time bounds are still achievable.

\subparagraph*{Equal elements.} The assumption that all elements of the matrix $A$ are distinct is done without loss of generality. If there are equal elements, we can instead consider the matrix $B$, whose elements $B_{ij} = (A_{ij},i,j)$ are to be compared lexicographically. This can be done implicitly, without the need of storing $B$. Notice that if $A_{ij}$ is a strict saddlepoint of $A$, then $B_{ij}$ is the (necessarily unique) strict saddlepoint of $B$. Thus, we can solve the problem on $A$ by finding a strict saddlepoint of $B$ (if there is one) and testing if it is indeed a strict saddlepoint of $A$.

\subparagraph*{Rectangular matrices.} We briefly discuss how the algorithm can be adapted to non-square matrices. Suppose $m > n$ and let $A$ be an $m \times n$ matrix (the case of an $n \times m$ matrix can be handled similarly). We divide $A$ into $\lceil m/n \rceil$ possibly overlapping $n \times n$ submatrices that fully cover $A$ and compute the strict saddlepoints of each submatrix (whenever it exists), in $O(m)$ total time; let $Q$ be the set of these \emph{local saddlepoints}. Either $Q$ is empty and then $A$ has no strict saddlepoint, or each row of $A$ must contain an entry larger or equal to each element of $Q$. In the latter case only the maximum of $Q$ can be a strict saddlepoint of $A$, and this can be verified in $O(m)$ time. 

\subparagraph*{Parallelization.} 
Finally, we remark on the changes necessary for an efficient parallel implementation as mentioned in \S\,\ref{sec1}. 
First, we set the size parameter of $\textsc{ReduceMatrix}$ to $s = (\log_2{n})^c$, for a sufficiently large constant $c$.
After reduction, we are then left with an $O(\polylog n)$-size matrix, which we can solve in $O(\polylog n)$ time and work (e.g., by a deterministic algorithm).
Note that by running more reduction steps than before, we have increased the probability of failure, which is now $e^{-\Omega((\log{n})^{c/20})} \leq 1/n^{\Omega(\polylog{n})}$, by Theorem~\ref{thm3}, so the algorithm still succeeds with high probability. 

The $O(\log{n})$ iterations of the main loop in $\textsc{ReduceMatrix}$ and the $O(\log{n})$ iterations of the Phase 1 loop in the $\textsc{FindPivot}$ calls can be invoked in sequence. Sampling independently from each row (in both Phase 1 and Phase 2), selection from each row (in Phase 2), and comparisons of an element with other elements of its row or column can be invoked fully in parallel, without increasing the total work. 

Two crucial components remain: (1) Selection from $n$ items can be implemented in parallel in $O(\log{n})$ time and $O(n)$ work~\cite{han2007optimal}, and (2) Array manipulation (compacting an $O(n)$-size array after deleting a constant fraction of row- or column-indices in $O(\log{n})$ time and $O(n)$ work) can be achieved by standard techniques based on prefix sums~\cite{blelloch1990prefix}. 

Overall, to find the strict saddlepoint of an $n\times n$ matrix with high probability, we need $O(\polylog n)$ parallel time, and the $O(n)$ bound on the total work from the sequential analysis continues to hold. 

\section{Lower bound against randomized algorithms}\label{sec5}

\begin{theorem}
Every randomized algorithm that decides if a given matrix $M$ has a (non-strict) saddlepoint and returns its value correctly with probability at least $\sfrac{5}{6}$ must query in expectation $\Omega(n^2)$ entries, for some $n \times n$ matrix $M$.

\end{theorem}
\begin{proof}%

    Consider a random $n\times n$ matrix $M$ generated by the following process:
    \begin{itemize}
        \item Start with all entries set to $0$.
        \item In every row, choose an entry uniformly at random and set it to $2$.
        \item Choose, uniformly at random, an entry with value $2$. Change it to $1$ or $-1$ with probability $\frac{1}{2}$ each.
    \end{itemize}
    Call $t$ the unique entry with value $1$ or $-1$. Notice that if $t=1$, then either there is no saddlepoint or the value of the saddlepoint is $1$ (the latter happens exactly if all $2$s are in the column of $t$). If $t=-1$, then the value of the saddlepoint is $0$ (pick any $0$ in $t$'s row). %
    Consider some arbitrary fixed (deterministic) algorithm that finds the saddlepoint. Observe that, unless the algorithm queries $t$, the probability of it succeeding is at most $\frac{1}{2}$. 
    
        Let us give the algorithm a budget of $n^2/1000$ queries, and argue that the probability that it succeeds within this budget is less than $\frac{2}{3}$. %
    Call the unique nonzero entry of each row a \emph{special element}, 
call rows with at least $n/10$ entries queried \emph{heavy}, and other rows \emph{light}; a row can change status at most once, from light to heavy.
Notice that the algorithm can reveal at most $\frac{n^2}{1000} / \frac{n}{10} = \frac{n}{100}$ special elements in heavy rows. 

All queries in light rows reveal a special element with probability at most $1/(9n/10) = \frac{10}{9n}$. (This is because the special element, unless already revealed, is equally likely to be in any of the unqueried places; otherwise, if the special element of the row has already been revealed, then the probability is zero.) The expected number of special elements revealed in light rows is thus at most $\frac{n^2}{1000} \cdot \frac{10}{9n} = \frac{n}{900}$. By Markov's inequality, the probability of revealing $k$ times as many special elements is at most $\frac{1}{k}$, e.g., the probability of revealing more than $\frac{n}{100}$ is at most $\frac{1}{9}$.
    
Thus, with probability at least $1-\frac{1}{9}$, the number of special elements revealed (in all rows) is at most $\frac{n}{100} + \frac{n}{100} = \frac{n}{50}$. Assume that, indeed, at most $\frac{n}{50}$ special elements are revealed. 
Then, the probability of $t$ being among these is at most $\frac{1}{50}$.

Overall, the probability that the algorithm succeeds is at most $\frac{1}{50} + \frac{1}{9} + \frac{1}{2} < \frac{2}{3}$. The first term is for the case when the algorithm reveals at most $n/50$ special elements and finds $t$, the second term is for the case when the algorithm reveals more than $n/50$ special elements, and the third term is for the case when the algorithm succeeds without finding $t$. 

  Let $T(n)$ denote the minimum expected number of queries of a deterministic algorithm for the saddlepoint problem which errs with probability at most $\frac{1}{3}$ on this distribution of inputs. The previous discussion shows that $T(n)\geq \Omega(n^2)$.

  Let $T'(n)$ denote the expected number of queries of any randomized algorithm for the saddlepoint problem which errs with probability at most $\frac{1}{6}$ on the worst-case input. By Yao's minimax principle, %
  e.g.,~\cite[Proposition 2.6]{motwani1995}, we have $T'(n) \geq \frac{1}{2}T(n)$. Thus, the claimed result holds.
  \end{proof}


\begin{thebibliography}{10}

\bibitem{bellare1994randomness}
Mihir Bellare and John Rompel.
\newblock Randomness-efficient oblivious sampling.
\newblock In {\em Proceedings 35th Annual Symposium on Foundations of Computer
  Science}, pages 276--287. IEEE, 1994.

\bibitem{Bienstock1991}
Daniel Bienstock, Fan Chung, Michael~L. Fredman, Alejandro~A. Sch\"{a}ffer,
  Peter~W. Shor, and Subhash Suri.
\newblock A note on finding a strict saddlepoint.
\newblock {\em Am. Math. Monthly}, 98(5):418–419, April 1991.
\newblock \href {https://doi.org/10.2307/2323858} {\path{doi:10.2307/2323858}}.

\bibitem{blelloch1990prefix}
Guy~E Blelloch.
\newblock Prefix sums and their applications.
\newblock 1990.

\bibitem{Byrne1991}
Christopher~C. Byrne and Leonid~N. Vaserstein.
\newblock An improved algorithm for finding saddlepoints of two-person zero-sum
  games.
\newblock {\em Int. J. Game Theory}, 20(2):149--159, June 1991.

\bibitem{chambolle2011first}
Antonin Chambolle and Thomas Pock.
\newblock A first-order primal-dual algorithm for convex problems with
  applications to imaging.
\newblock {\em Journal of mathematical imaging and vision}, 40:120--145, 2011.

\bibitem{chambolle2016ergodic}
Antonin Chambolle and Thomas Pock.
\newblock On the ergodic convergence rates of a first-order primal--dual
  algorithm.
\newblock {\em Mathematical Programming}, 159(1-2):253--287, 2016.

\bibitem{dallant2024finding}
Justin Dallant, Frederik Haagensen, Riko Jacob, L{\'a}szl{\'o} Kozma, and
  Sebastian Wild.
\newblock Finding the saddlepoint faster than sorting.
\newblock In {\em 2024 Symposium on Simplicity in Algorithms (SOSA)}, pages
  168--178. SIAM, 2024.

\bibitem{dubhashi2009concentration}
Devdatt~P. Dubhashi and Alessandro Panconesi.
\newblock {\em Concentration of measure for the analysis of randomized
  algorithms}.
\newblock Cambridge University Press, 2009.

\bibitem{han2007optimal}
Yijie Han.
\newblock Optimal parallel selection.
\newblock {\em ACM Transactions on Algorithms (TALG)}, 3(4):38--es, 2007.

\bibitem{Knuth1}
Donald~E. Knuth.
\newblock {\em The Art of Computer Programming, Volume 1 (3rd Ed.): Fundamental
  Algorithms}.
\newblock Addison Wesley Longman Publishing Co., Inc., USA, 1997.

\bibitem{Knuth2}
Donald~E. Knuth.
\newblock {\em The Art of Computer Programming, Volume 2 (3rd Ed.):
  Seminumerical Algorithms}.
\newblock Addison Wesley Longman Publishing Co., Inc., USA, 1997.

\bibitem{lin2020near}
Tianyi Lin, Chi Jin, and Michael~I. Jordan.
\newblock Near-optimal algorithms for minimax optimization.
\newblock In {\em Conference on Learning Theory}, pages 2738--2779. PMLR, 2020.

\bibitem{Llewellyn1988}
Donna~Crystal Llewellyn, Craig Tovey, and Michael Trick.
\newblock Finding saddlepoints of two-person, zero sum games.
\newblock {\em The American Mathematical Monthly}, 95(10):912--918, 1988.
\newblock \href {https://doi.org/10.1080/00029890.1988.11972116}
  {\path{doi:10.1080/00029890.1988.11972116}}.

\bibitem{gt}
Michael Maschler, Shmuel Zamir, and Eilon Solan.
\newblock {\em Game theory}.
\newblock Cambridge University Press, 2020.

\bibitem{motwani1995}
Rajeev Motwani and Prabhakar Raghavan.
\newblock {\em Randomized algorithms}.
\newblock Cambridge University Press, 1995.

\bibitem{nedic2009subgradient}
Angelia Nedi{\'c} and Asuman Ozdaglar.
\newblock Subgradient methods for saddle-point problems.
\newblock {\em Journal of optimization theory and applications}, 142:205--228,
  2009.

\bibitem{razaviyayn2020nonconvex}
Meisam Razaviyayn, Tianjian Huang, Songtao Lu, Maher Nouiehed, Maziar Sanjabi,
  and Mingyi Hong.
\newblock Nonconvex min-max optimization: Applications, challenges, and recent
  theoretical advances.
\newblock {\em IEEE Signal Processing Magazine}, 37(5):55--66, 2020.

\bibitem{vadhan2012pseudorandomness}
Salil~P. Vadhan.
\newblock Pseudorandomness.
\newblock {\em Foundations and Trends{\textregistered} in Theoretical Computer
  Science}, 7(1--3):1--336, 2012.

\end{thebibliography}
\end{document}